\documentclass[11pt]{article}

\usepackage{dsfont}

\usepackage{tikz}
\usepackage{pgfplots}
\pgfplotsset{compat = newest}
\usetikzlibrary{decorations.pathreplacing,calligraphy}

\usepackage{amssymb,amsmath,amsthm}
\usepackage{bbm}
\usepackage[noend]{algpseudocode}
\usepackage{graphicx}
\usepackage{verbatim}
\usepackage{url,xspace,hyperref}
\usepackage{makecell}
\usepackage{cite}
\usepackage{caption}
\usepackage{subcaption}
\usepackage{multirow}
\usepackage{multicol}
\usepackage{latexsym}
\usepackage{amsmath,amssymb,enumerate}
\usepackage{xspace}
\usepackage{algorithm,algorithmicx}
\usepackage{float}
\usepackage{xcolor}
\usepackage{mathrsfs}
\usepackage{cleveref}
\usepackage{cancel}

\usepackage{fullpage}
\usepackage{geometry}
\usepackage{setspace}

\newcommand{\ignore}[1]{}

\ignore{ 

\documentclass[11pt]{article}

\usepackage{dsfont}

\usepackage{tikz}
\usepackage{pgfplots}
\pgfplotsset{compat = newest}
\usetikzlibrary{decorations.pathreplacing,calligraphy}
\usepackage{authblk}
\usepackage{amssymb,amsmath,amsthm}
\usepackage{bbm}
\usepackage[noend]{algpseudocode}
\usepackage{graphicx}
\usepackage{verbatim}
\usepackage{url,xspace,hyperref}
\usepackage{makecell}
\usepackage{cite}
\usepackage{caption}
\usepackage{subcaption}
\usepackage{multirow}
\usepackage{multicol}
\usepackage{latexsym}
\usepackage{amsmath,amssymb,enumerate}
\usepackage{xspace}
\usepackage{algorithm,algorithmicx}
\usepackage{float}
\usepackage{xcolor}
\usepackage{mathrsfs}
\usepackage{cleveref}
\usepackage{cancel}

\usepackage{fullpage}
\usepackage{geometry}
\geometry{letterpaper,tmargin=1in,bmargin=1in,lmargin=1in,rmargin=1in}
\usepackage{setspace}
}

\newtheorem{theorem}{Theorem}[section]
\newtheorem{definition}{Definition}[section]

\newtheorem{lemma}[theorem]{Lemma}

\newtheorem*{fact*}{Fact}
\newtheorem*{remark*}{Remark}
\def\odt{\ensuremath{{\sf ODT}}\xspace}
\def\opt{\ensuremath{{\sf OPT}}\xspace}
\def\alg{\ensuremath{{\sf ALG}}\xspace}
\def\st{\ensuremath{{\sf Score}}\xspace}
\def\start{\ensuremath{{\sf start}}}
\def\endd{\ensuremath{{\sf end}}}
\def\E{\ensuremath{\mathbb{E}}}
\def\Stem(w){\mathsf{Stem}(w)}

\title{A Simple Approximation Algorithm for  Optimal Decision Tree}

\author{Zhengjia Zhuo\thanks{Department of Industrial and Operations Engineering, University of Michigan, Ann Arbor, USA. Research supported in part by NSF grant  CCF-2418495.} \and Viswanath Nagarajan$^*$}

\hypersetup{
pdftitle={A Simple Approximation Algorithm for  Optimal Decision Tree},
pdfsubject={q-bio.NC, q-bio.QM},
pdfauthor={David S.~|H|ippocampus, Elias D.~Striatum},
pdfkeywords={First keyword, Second keyword, More},
}

\begin{document}

\maketitle
\begin{abstract}
    Optimal decision tree (\odt) is a fundamental problem arising in applications such as active learning, entity identification, and medical diagnosis.
An instance of \odt is given by $m$ hypotheses, out of which an unknown ``true'' hypothesis is drawn according to some probability distribution. 
An algorithm needs to identify the true hypothesis by making queries: each query incurs a cost and has a known response for each hypothesis. 
The goal is to minimize the expected query cost to identify the true hypothesis. We consider the most general setting with arbitrary costs, probabilities and responses.
\odt is NP-hard to approximate better than $\ln m$ and there are $O(\ln m)$ approximation algorithms known for it.
However, these algorithms and/or their analyses are quite complex. Moreover, the leading constant factors are large. 
We provide a simple algorithm and analysis for \odt, proving an approximation ratio of $8 \ln m$. 
\end{abstract}

\section{Introduction}
Identifying a  target hypothesis from a pool of candidates using sequential queries is a  fundamental task in many applications~\cite{Moret82,Murthy98}. We briefly mention three examples here. 
In medical diagnosis, when faced with a new patient 
 one wants to identify the disease inflicting them  using medical tests at low cost. 
In species identification, one wants to identify the species of an entity by checking the fewest number of characteristics. In active learning, given  data points with associated labels and a set of classifiers  (each of which maps  data points to labels), one wants to identify the correct classifier by querying the actual labels at the fewest number of points.  In each of these applications, it is also  natural to assume an {\em a priori} distribution for the target hypothesis, so that the goal is to minimize the {\em expected} cost of identification. 

The optimal decision tree (\odt) problem is used to model  such identification tasks and has been  widely-studied for over 50 years; see e.g.,~\cite{GareyG74,Loveland85,Kosaraju1999,Dasgupta2004,ChakaravarthyPRS09,GuilloryB09,GuptaNR17,Ray2019}. \odt was shown to be NP-hard in \cite{HyafilR76} and it is also NP-hard to approximate to a factor better than $\ln m$, where $m$ is the number of hypotheses~\cite{ChakaravarthyPRS09}. Several different approximation algorithms have been designed for \odt, resulting in an asymptotically optimal $O(\ln m)$-approximation even in the most general setting with non-uniform probabilities, non-uniform query costs and an arbitrary number of possible query-responses~\cite{GuptaNR17,CicaleseLS17,NavidiKN20}.   However, these algorithms/analyses are quite complex and the leading constants in the approximation ratios are very large. Our main result is a simple algorithm and analysis for the general \odt problem (with arbitrary probabilities, costs and responses) showing an approximation ratio of $8\cdot \ln m$.

\def\sse{\subseteq}
\subsection{Problem Definition}
In the optimal decision tree  (\odt) problem, we have a finite hypothesis class $H$ containing $m$ {\em hypotheses}. An unknown target hypothesis $h^* \in H$  is drawn from a (known) probability distribution $D=\{p_h\}_{h\in H}$, where $p_h=\Pr[h^*=h]$ and $\sum_{h\in H} p_h=1$. There is also a set $E$ of {\em queries} and a set $R$ of {\em responses}. Each query $e\in E$ has a cost $c(e)\ge 0$ and maps any hypothesis $h\in H$ to a  response $e(h) \in R$. The table of responses $e(h)$ for all queries $e\in E$ and hypotheses $h\in H$ is known upfront. The goal is to identify the target $h^*$ using sequential queries. 
When a query $e$ is performed, we observe the response $e(h^*)$, which can be used to eliminate all hypotheses $h\in H$ with $e(h)\ne e(h^*)$.  The objective is to minimize  the expected cost to identify $h^*$.

A deterministic solution (or policy) for \odt corresponds to a decision tree $\pi$ where each internal node $w$ is labeled by a query $\pi_w\in E$ and each  branch is labeled by a  response. The decision tree execution starts at the root of $\pi$  and proceeds as follows: at any internal node $w$, we  perform   query $\pi_w$,  observe its response $r$ and move along the branch labeled $r$ to the next node.  Note that each hypothesis $h\in H$ traces a unique  path  in $\pi$ (when all responses are according to $h$), ending at some leaf node $\ell(\pi,h)$. The decision tree $\pi$ is said to be  feasible   if the leaf nodes $\ell(\pi,h)$ are distinct for each $h\in H$. So there is a one to one correspondence between leaf nodes in the decision tree and the hypothesis in $H$. The expected cost of $\pi$ is $\sum_{h\in H} p_h\cdot C(\pi,h)$ where $C(\pi,h)$ is the total cost of queries on the path in $\pi$ traced by $h$. 

Optimal decision tree  can also be formulated as a Markov Decision Process (MDP) with state-space $2^{E\times R}$ (corresponding to previously observed  query-response pairs) and actions $E$ (corresponding to the next query to perform). While this view is not useful in coming up with efficient algorithms  (as the state space is exponential), it may be useful conceptually. For example, this shows that there is always a deterministic optimal policy for \odt: so we can just focus on deterministic policies (as defined above).

It will also be convenient to associate a {\em state} with  any node $w$ of decision tree $\pi$, that contains all 
query-response pairs on the path in $\pi$ from  root to $w$. That is, if the root-$w$ path has  queries  $e_1,\cdots e_k$  with (respective) responses $r_1,\cdots r_k$ then the state at $w$ is $\{(e_i,r_i) : 1\le i\le k\}$.
Note that each node in $\pi$ has a distinct state. To reduce notation, when it is clear from context, we use $w$ to also denote the state at node $w$.   For any state $w=\{(e_i,r_i) : i\in [k]\}$, let $V(w)\sse H$ denote all the hypotheses that are compatible with this state, i.e., $V(w) = \{h\in H : e_i(h) = r_i , \forall i\in [k] \}$. In the active learning literature~\cite{Dasgupta2004}, $V(w)$ is also called the ``version space''. At state $w$, all the hypotheses in $H\setminus V(w)$ are said to be {\em eliminated}. Note that we have $|V(w)|=1$ for any {\em leaf node} $w$ in decision tree $\pi$.  

For any subset $S\sse H$, we use $p(S):= \sum_{h\in S} p_h$ to denote the total probability in $S$. We use $p_{min}=\min_{h\in H} p_h$ to denote the minimum probability of a hypothesis. Similarly, $c_{max}=\max_{e\in E} c(e)$ and $c_{min}=\min_{e\in E} c(e)$ denote the maximum/minimum query costs.

\paragraph{Greedy policies.} Several algorithms for \odt (including ours) are greedy and involve selecting one ``maximum benefit'' query at each step. Specifically, at any state $w$ of a greedy policy, we compute a numeric value for each query $e\in E$ (that depends only on $w$ and $e$) and select the query that maximizes this value. For any state $w$ and query $e$, let $\{S_i(e,w) : i\in R\} $ denote the partition of $V(w)$ based on the outcome of query $e$, i.e., $S_i(e,w) = \{h\in V(w) : e(h)=i\}$ for all $i\in R$; also define $s_i(e,w):=\frac{p(S_i(e,w))}{p(V(w))}$ the conditional probability of observing response $i$ upon query $e$. The choice of the greedy criterion greatly influences the approximation ratio achieved as well as its analysis. Indeed, several different greedy criteria for \odt have been used/analyzed in prior work. Table~\ref{Approx Rate} summarizes these results; to reduce clutter, we drop  the dependence on state $w$ and query   $e$ in the greedy criteria.

{\small 
\begin{table}[t]
\centering
    \renewcommand\arraystretch{1.5}
\begin{tabular}{|c|c|c|c|c|c|}
\hline & $|R|>2$  & NU cost & NU prob. & Ratio & Greedy criterion\\
\hline \text {\cite{Kosaraju1999} } & $\times$ & $\times$ & $\checkmark$ & $O(\log m)$& $ \min(s_1,s_2) $ \\
\text { \cite{Dasgupta2004}} & $\times$ & $\times$ & $\checkmark$ & $4 \,  \ln \frac{1}{p_{min}} $& $  \min(s_1,s_2)$   \\
\text {  \cite{Adler2008}} & $\times$ & $\checkmark$ & $\times$ & $\ln m$ &$ \frac{1}{c}  |S_1| |S_2| $ \\
\text { \cite{Chaka2011} } & $\checkmark $& $\times$ &$ \checkmark $& $O(\log |R| \log m)$ & $    \sum_{i\neq j} s_i s_j $ \\
\text { \cite{ChakaravarthyPRS09} } & $\checkmark $& $\times$ &$ \times$& $5.8\, \ln m$ & $   |V| - \max_{i} |S_i|$ \\
\text {  \cite{GuilloryB09}} & $\checkmark$ & $\checkmark$ & $\checkmark$ &$ 108 \, \ln \left( m \cdot \frac{c_{max}}{c_{min}}\right)$ & $ \frac{1}{c} \left(1-\sum_{i=1}^q s_i^2\right)$\\
\text { \cite{GuptaNR17,CicaleseLS17} } & $\checkmark$  &$\checkmark$ &$\checkmark$ & $O(\log m)$& not greedy \\
\text { \cite{NavidiKN20} } &  $\checkmark$  &$\checkmark$ &$\checkmark$ & $O(\log m)$ & $\frac{1}{c} \left(|V| (1-s_1) + \sum_{i} |S_i|(1-s_i)\right)$  \\
\cite{Ray2019} & $\times$ & $\times$ & $\times$ & $\frac{(4+\epsilon)\ln m}{\ln\log m}$& $\min(s_1,s_2)$  \\
\text { \cite{EsfandiariKM21} } &  $\checkmark$   &$\times$ &$\checkmark$& $  2\, \ln \frac{1}{p_{min}} $ & $\sum_{ i} |S_i|(1-s_i)$ \\
\text { \cite{althani2024}} &  $\checkmark$   &$\checkmark$ &$\times$& $ 4 \, \ln m $ & $\frac{1}{c} \sum_{ i} |S_i|(|V|-|S_i|)$ \\
\hline
\text{\bf This paper} & $\checkmark$ &$ \checkmark$ &$ \checkmark$  & $\mathbf{  8 \ln m }$ & $\frac{1}{c} \sum_{i} |S_i|(1-s_i)$ \\
\hline
\end{tabular}
    \caption{Approximation ratios summary.  Here, NU denotes ``non-uniform''. }
    \label{Approx Rate}
    \vspace{5mm}
\end{table}
}

 \subsection{Related Work}
Early work~\cite{GareyG74,Loveland85} on the \odt problem focused on the special case of binary responses ($|R|=2$) and analyzed the performance of a  greedy policy in special instances.  The first approximation algorithm for  \odt was by \cite{Kosaraju1999}, who proved an $O(\ln m)$-approximation in the case of uniform costs and  binary responses. \cite{Dasgupta2004} showed how \odt is applicable to   active learning and re-discovered the same greedy policy as \cite{Kosaraju1999}, again for uniform  costs and  binary responses. \cite{Adler2008} analyzed a different greedy policy and proved a $\ln m$ approximation ratio for uniform probabilities and binary responses (via a simple proof).

\odt with non-binary responses ($|R|>2$) has also been studied, starting with \cite{Chaka2011} who obtained an $O(\log |R| \log m)$-approximation for uniform costs. They also proved that \odt (with non-uniform probabilities) cannot be approximated better than $\ln m$ and \odt with uniform probabilities cannot be approximated better than factor $4$. Then,  \cite{GuilloryB09} considered the most general case of \odt (with arbitrary costs, probability and responses) and obtained an $O(\ln \frac{m c_{max}}{c_{min}})$-approximation ratio.\footnote{Henceforth, ``general \odt'' will refer to the problem  with arbitrary costs, probability and responses. } \cite{GuptaNR17} obtained the first  $O(\ln m)$-approximation for the general \odt problem; their algorithm  was somewhat simplified in \cite{CicaleseLS17}. We note that both these algorithms were non greedy and involved carefully selecting a sequence of queries in each step. Later, \cite{NavidiKN20} obtained a greedy policy for  general \odt which was also shown to be  an $O(\ln m)$-approximation. Still, the analysis for general \odt was  complex in all three papers \cite{GuptaNR17,CicaleseLS17,NavidiKN20}  and the leading constant factors were large. (While the leading constant was not explicitly calculated in \cite{GuptaNR17,CicaleseLS17,NavidiKN20}, we  estimate  that these factors are over $100$.) 

A different  approach for \odt is via the framework of adaptive submodularity~\cite{GK17}. In particular,  recent results by \cite{EsfandiariKM21} and \cite{althani2024} imply approximation ratios of  $2\ln(1/p_{min})$ (for uniform costs) and $4\ln m$ (for uniform probability). While the resulting algorithms are simple greedy ones, a  full analysis  requires multiple steps: defining  adaptive submodularity, proving that the goal in \odt can be expressed as covering an adaptive-submodular function  and minimum-cost cover  in   adaptive submodularity. Moreover, the best ratio for general \odt that one can expect via this approach is $\ln(1/p_{min})$, which could be arbitrarily large as a function of  $m$. This is because the adaptive-submodular function corresponding to \odt has a ``goal value'' of $\frac{1}{p_{min}}$ \cite{GK17}.

For the special case of \odt with uniform cost and probabilities and binary response, \cite{Ray2019} showed that the natural greedy policy achieves a sub-logarithmic approximation ratio of $O(\frac{\ln m}{\ln\ln m})$. Their result also extends to a $\left(\frac{12\ln (1/p_{min})}{\ln \opt}+\ln(1/p_{min})\right)$ approximation for general responses and probabilities (still unit costs).    

\subsection{Main Result and Techniques}

We analyze a natural greedy policy for general \odt.  Observe that  the goal in \odt is to  eliminate $m-1$ hypotheses, so that the only remaining hypothesis is  $h^*$. Moreover, the number of newly eliminated hypotheses due to any query is random: so it is natural to consider the {\em expectation} of this quantity as a measure of progress.  This motivates our greedy criterion: select  query $e$ that maximizes the ratio of expected number of newly eliminated hypotheses to its cost $c(e)$.  We prove:
\begin{theorem}\label{thm:main}
Our greedy policy for general \odt has approximation ratio at most $8\cdot (1+\ln m)$.  
\end{theorem}

Our high-level proof structure is the same as in \cite{althani2024} for adaptive-submodular cover. Specifically, we view the expected cost of the greedy and optimal policies as integrals of ``non-completion'' probabilities. So, it suffices to relate the non-completion probability in greedy at cost $t$ to that of the optimum at a scaled  cost $t/L$ (for some approximation parameter $L$).   In order to relate these non-completion probabilities, we consider the integral of the greedy criterion value and prove  upper  and lower bounds on this quantity. The upper bound proof is basically the same as in \cite{althani2024}, which is reminiscent of the analysis for (deterministic) set cover. Our main technical  contribution is in the lower bound proof (Lemma~\ref{lemma-lower bnd}). 
To this end, for any state $w$ of the greedy policy (at cost $t$),  we identify a particular sequence  of queries in the optimal policy (called  $\Stem(w)$)  having total cost at most $t/L$ such that the expected  number of newly eliminated hypotheses (added over the queries in  $\Stem(w)$) is large. 

\section{\odt Greedy Algorithm}
In this section we prove Theorem~\ref{thm:main}.  
For any state $w\sse E\times R$, we use the following:
\begin{itemize}
    \item Recall that  $V(w)\sse H$ is the set of compatible hypotheses at state $w$. 
    \item $D|w$  denotes the conditional distribution of target hypothesis $h^*$, with $\Pr_{D|w}[h^*=h]=\frac{p_h}{p(V(w))}$ for all $h\in V(w)$ and $\Pr_{D|w}[h^*=h]=0$ if $h\not\in V(w)$.
    \item For any query $e\in E$, recall that $\{S_i(e,w) : i\in R\}$ is the partition of $V(w)$ based on $e$'s response and $s_i(e,w) = \Pr_{D|w}[h^*\in S_i(e,w)]$. 
    \item For any query $e\in E$, define $\Delta(e|w)$ as the expected number of newly eliminated hypotheses when query $e$ is performed at state $w$. Formally,   
  \begin{equation} \label{eq:delta def}
          \Delta(e|w)= E_{h^* \sim D|w}\left[  V(w)  - V(w\cup (e,e(h^*)) )   \right] = \sum_{ i\in R} s_i(e,w)\cdot \left(V(w) - S_i(e,w)\right) . 
  \end{equation}
    Throughout this paper, for any set $S\sse H$, when the context is clear, we also use $S$ to represent its cardinality $|S|$.

\end{itemize}

 Our greedy policy is described formally in Algorithm~\ref{alg:1} and is initialized with Greedy($\emptyset$).   While we have described the policy  as a recursive procedure, it is easy to build an explicit decision tree by running the policy under all choices of $h^*\in H$.

\begin{algorithm}[H]
\caption{Greedy Algorithm for optimal decision tree}\label{alg:1}
\begin{algorithmic}[1]

\Procedure{Greedy}{$w$}
\If{$|V(w)|=1$} 
\State Return the unique hypothesis in  $V(w)$ as $h^*$.
\Else
\State Choose $e\in E$ to maximize $\frac{\Delta(e|w)}{c(e)}$
\State Perform query $e$ to get  response $r$
\State Recurse on Greedy($w\cup \{(e,r)\}$)
\EndIf
\EndProcedure
\end{algorithmic}
\end{algorithm}

For convenience, we use ``time'' to refer to the cumulative cost of any policy. Let \alg  denote  the total cost incurred by our greedy policy, so the greedy objective is $\E[\alg]$. Similarly, let \opt denote the cost of the optimal policy. We also define the ``non-completion'' probabilities:
$$
a_t: = \Pr[\alg \geq t] \quad\mbox{and} \quad o_t:= \Pr[\opt  \geq t], \quad \forall t\ge 0.
$$

Clearly, we have

\begin{equation}\label{eq:ALG and at}
\E[\alg] = \int_{t\ge 0} a_t dt \quad \mbox{and} \quad  \E[\opt]= \int_{t\ge 0} o_t dt .    
\end{equation}

\begin{definition}\label{Def:Time}
Consider any state $w = \{(e_j,r_j) : j\in [k]\}$ in the greedy policy. We define its starting time $\start(w) := \sum_{j=1}^k  c(e_j)$. If $w$ corresponds to an internal node,  let $e$ denote the query labeling this node and define the end time $\endd(w) = \start(w)+c(e)$.  
\end{definition}

Note that the start time of $w$ is the earliest time when all the information  in $w$ becomes available, and its end time is the earliest time when some additional information becomes available (as the result of the next query).
For the root node we have $\start(\emptyset)=0$ and for any leaf node $w$ we have $\endd(w)=\start(w)$ because there is no query at the leaf node. See Figure~\ref{fig:example} for an example.

\begin{figure}[H]
    \centering
    \includegraphics[width=0.5\linewidth]{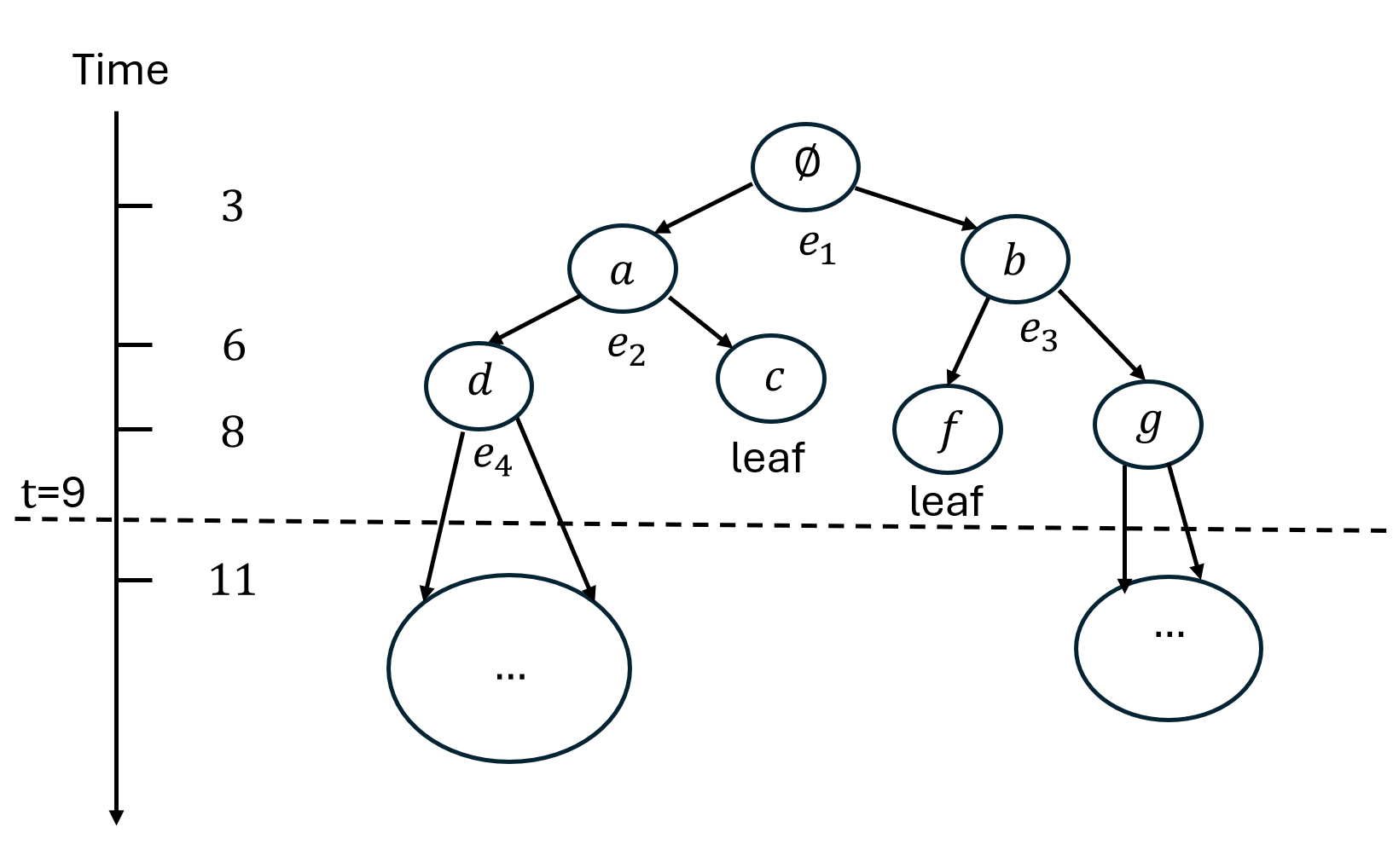}
    \caption{An example policy (decision tree).  The initial state is   $\emptyset$ and the rest are labeled  $a$-$g$. The costs $c(e_1)=c(e_2)=3$ and $c(e_3)=c(e_4) =5$. We have  $\start(b) = c(e_1)=3$ and $\endd(b)=\start(b)+c(e_3)=8$. Similarly, $\start(d) = c(e_1)+ c(e_2)=6$  and $\endd(d) = \start(d)+ c(e_4)=11$. The active states at time $t=9$ are $\{d,g\}$. }
    \label{fig:example}
\end{figure}

\def\act{\Gamma}
\begin{definition}[Active states]\label{Def:Sigma_t}
    For any time $t \geq 0$, we define the set of active states in the greedy policy at time $t$ as $
    \act_t: = \{w  | \start(w) \leq t \leq \endd(w) \} $. 
\end{definition}
Observe that $\act_t$ consists of all possible states that the greedy policy may  be in at time $t$ (if it has not already terminated). 
The following lemma is immediate from the definition of active states. For notational convenience, we use $V_w$ to denote $V(w)$ in the rest of this paper. 
\begin{lemma}\label{partition}
    For any time $t$, the sets $\{V_w : w\in \act_t\}$ are   disjoint and  
    $$\sum_{w\in \act_t}p(V_w) =\Pr(\alg \geq t)= a_t .$$
\end{lemma}

\begin{definition}[Score] \label{score}For any time $t$ and active state $w\in \act_t$, we define 
    $$
\mathrm{Score}(t,w) = \max_{e\in E} \frac{\Delta(e|w)}{c(e) ( V_w -1)} .
$$

Moreover, for any time $t$, define the random quantity 

$$\st_t =\left\{ \begin{array}{ll}
     0 & \mbox{if greedy terminates before time $t$} \\
     \mathrm{Score}(t,w) & \mbox{where $w\in \act_t$ is the state at time $t$ in greedy} 
\end{array} \right. . $$
\end{definition}
Note that the score is  a scaled version of the greedy criterion.  
Our main technical result is the following lower bound. 

\begin{lemma}\label{lemma-lower bnd}
For any time $t$ and $L\ge 0$, we have  
$$
\E_{h^*\sim D }[\st_t] \geq   \frac{L}{2t} (a_t - o_{t/L}) . 
$$
\end{lemma}

Before proving this lemma, we introduce some intermediate results. Below, we fix time $t$ and $L$ as in the lemma.

\begin{definition}[Heavy part]\label{def:heavy part}
Fix a state $w\in \act_t$ and a query $e$. Let $B_w^e$ denote the part in $\{S_i(e,w):i\in R\}$ with the largest size. We call $B_w^e$ the heavy part of $V_w$ with respect to $e$; also define $J_w^e = V_w \setminus B_w^e$.
\end{definition}

\def\Stem(w){\mathsf{Stem}(w)}

In the following, we will associate an active state $w$ with a particular path $\Stem(w)$ in the optimal decision tree. We will use this association to argue that our greedy policy is competitive with the optimal policy.

\paragraph{Finding a path in the optimal tree.} Given an active state $w \in \act_t$, we construct a path $\Stem(w)$ in the optimal tree starting at the root as follows. At any node labeled with query $e$, we follow the branch corresponding to the heavy part $B_w^e$. This process continues as long as the total cost of queries on the path remains \emph{at most} $\frac{t}{L}$. With a slight abuse of notation, let $\Stem(w)$ denote the set of queries along this path. By construction, we have the total cost in $\Stem(w)$ is:
\begin{equation}\label{eq:I}
\sum_{e \in \Stem(w)} c(e) \leq \frac{t}{L}.
\end{equation}

\begin{lemma}\label{score low bnd}
For any state $w \in \act_t$ and $e \in \Stem(w)$, we have
$$
\frac{\Delta(e|w)}{V_w-1} \,\,\geq\,\, \frac{p(J_w^e)}{2 p(V_w)} \,+\, \frac{p(B_w^e)}{p(V_w)}\frac{J_w^e}{V_w-1}.
$$     
\end{lemma}
\begin{proof}
Fix $w\in \act_t$ and $e \in \Stem(w)$. Suppose that $e$ partitions $V_w$ into $q$ disjoint parts, i.e.,
$V_w = \bigcup_{i=1}^q S_i$. 
Without loss of generality, assume that $S_1$ has the largest size. So we have $B_w^e = S_1$ and $J_w^e = \bigcup_{i\geq2}S_i$. Then,
\begin{equation}\label{eq:half_V}
    V_w - S_i \stackrel{}{\geq} V_w - \frac{S_1 + S_i}{2} \stackrel{}{\geq} V_w - \frac{V_w}{2} = \frac{V_w}{2}, \quad \forall i \geq 2.
\end{equation}

Above, the first inequality follows because $S_1 \geq S_i$ (i.e., $S_1$ has the largest size). The second inequality follows because $S_1$ and $S_i$ are disjoint, which implies $S_1 + S_i \leq V_w$.
Now, we have

\begin{align}
\frac{\Delta(e|w)}{V_w-1} &\stackrel{}{=} \frac{1}{V_w-1}\sum_{ i=1}^q  \frac{p(S_i)}{p(V_w)}\left(V_w - S_i\right) \label{eq:a} \\
&= \frac{p(B_w^e)}{p(V_w)} \frac{V_w - S_1}{V_w-1} + \frac{1}{V_w-1}\sum_{i=2}^q  \frac{p(S_i)}{p(V_w)} \left(V_w - S_i\right) \nonumber \\
&\stackrel{}{\geq} \frac{p(B_w^e)}{p(V_w)} \frac{V_w - S_1}{V_w-1} + \frac{V_w}{2(V_w-1)} \sum_{i=2}^q  \frac{p(S_i)}{p(V_w)} \label{eq:b}  
 \stackrel{}{=} \frac{p(B_w^e)}{p(V_w)} \frac{J_w^e}{V_w-1} + \frac{p(J_w^e)}{2p(V_w)}. 
\end{align}
The equality in \eqref{eq:a} uses the definition of $\Delta(e|w)$ from Equation~\eqref{eq:delta def}. The inequality in \eqref{eq:b} follows from Equation~\eqref{eq:half_V}. Finally, the equality in \eqref{eq:b} is due to the definition $J_w^e = \bigcup_{i\geq2}S_i = V_w \setminus S_1$. 
\end{proof}

\begin{definition}\label{terminalstate}
For any active state $w\in\act_t$, we define:
\begin{equation}\label{V and I} 
  I_w := V_w \backslash \bigcup_{e \in \Stem(w)} J_w^e.
\end{equation}
 We further partition the  active states $\act_t$ based on this. Define 
    \begin{equation}\label{def:pos act}
            \act_t^-  := \{ w \in \act_t | \, |I_w| \leq 1 \} \quad \text{and} \quad \act_t^+ := \act_t \backslash \act_t^- = \{w \in \act_t | \, |I_w| \geq 2 \}. 
    \end{equation}
 
\end{definition}

In other words, $I_w\sse V_w$ are those hypotheses in $V_w$ that are compatible at the end of $\Stem(w)$ in the optimal policy. Moreover, if $w\in \act_t^+$ then the optimal policy does {\em not} identify any $h\in I_w$ before time $\frac{t}{L}$. (For $w\in \act_t^-$ there is a unique hypothesis  $h\in I_w$ which means that the optimal policy   identifies $h$  before time $\frac{t}{L}$.)  
 In particular, only the hypotheses corresponding to $\act_t^+$ help us establish a lower bound on the optimal cost. We formalize this in the following lemma.

\begin{lemma}\label{lemma:I_w}
For any $w \in \act_t^+$, we have
$$\sum_{w \in \act_t^+}p(I_w) \leq \Pr\left[\opt \geq \frac{t}{L}\right] = o_{t/L}.$$
\end{lemma}

\begin{proof}
We can bound the sum of probabilities as follows:
$$
\sum_{w \in \act_t^+}p(I_w) \stackrel{}{=} p\left(\bigcup_{w \in \act_t^+} I_w\right) \stackrel{}{\leq} \Pr\left(\opt \geq \frac{t}{L}\right) = o_{\frac{t}{L}}.
$$ 

For the first equality, note that the sets $\{I_w\}_{w \in \act_t}$ are disjoint since $I_w \subseteq V_w$ for each $w \in \act_t$ (Definition~\ref{terminalstate}) and the sets $\{V_w\}_{w \in \act_t}$ are disjoint by Lemma~\ref{partition}. For the inequality, we use that none of the hypotheses in $\bigcup_{w\in \act^+_t} I_w$ is identified before time $t/L$ in the optimal policy.  
\end{proof} 

\begin{lemma}\label{keylemma 1}
For any time $t$, 
$$
\mathrm{Score}(t,w) \geq
\begin{cases}
\frac{L}{2t}\frac{p(V_w) - p(I_w)}{p(V_w)} & \text{if } w \in \act_t^+ \\
\frac{L}{2t} & \text{if } w \in \act_t^- 
\end{cases} .
$$
\end{lemma}

\begin{proof}  

Recall that $\mathrm{Score}(t,w) = \max\limits_{e\in E}\frac{\Delta(e|w)}{(V_w-1) c(e)}$ from the greedy criterion. We will prove the lower bound for the $\mathrm{Score}(t,w)$ by an averaging argument. That is 
$$
\mathrm{Score}(t,w) \geq \max_{e\in \Stem(w)} \frac{\Delta(e|w)}{c(e) (V_w-1) } \ge   \frac{1}{\sum\limits_{e\in \Stem(w)} c(e)}\sum_{e\in \Stem(w)}\frac{\Delta(e|w)}{V_w-1} \geq \frac{L}{t}\sum_{e\in \Stem(w)}\frac{\Delta(e|w)}{V_w-1}.
$$
The last inequality uses equation~\eqref{eq:I}.
It now suffices to show that 
\begin{equation}\label{eq:4}
\sum_{e\in \Stem(w)}\frac{\Delta(e|w)}{V_w-1} \geq
\begin{cases}
\frac{p(V_w) - p(I_w)}{2p(V_w)} & \text{if } w \in \act_t^+ \\
\frac{1}{2} & \text{if } w \in \act_t^-
\end{cases}.     
\end{equation}
We have
\begin{align}
\sum_{e\in \Stem(w)}\frac{\Delta(e|w)}{V_w-1} &\stackrel{}{\geq} \sum_{e \in \Stem(w)} \left(\frac{p(J_w^e)}{2 p(V_w)} + \frac{p(B_w^e)}{p(V_w)}\frac{J_w^e}{V_w-1}\right) \label{eq:1}\\
&\stackrel{}{\geq} \frac{p(V_w) - p(I_w)}{2p(V_w)} + \sum_{e \in \Stem(w)}\frac{p(B_w^e)}{p(V_w)} \frac{J_w^e}{V_w-1} \label{eq:2} \\
&\stackrel{}{\geq} \frac{p(V_w) - p(I_w)}{2p(V_w)} + \sum_{e \in \Stem(w)}\frac{p(I_w)}{p(V_w)} \frac{J_w^e}{V_w-1} := RHS 
 \label{eq:3}.
\end{align}

The inequality in \eqref{eq:1} follows from Lemma~\ref{score low bnd}. For the inequality in \eqref{eq:2}, we use the definition of $I_w$, see~\eqref{V and I}, which implies that
$$
\sum_{e \in \Stem(w)} p(J_w^e) \geq p(V_w) - p(I_w).
$$ 
The inequality in \eqref{eq:3} holds because $I_w \subseteq B_w^e$ for any $w \in \act_t$ and $e \in \Stem(w)$. 

For any $w \in \act_t^-$, we have $|I_w| \leq 1$, which implies $\sum\limits_{e \in \Stem(w)} J_w^e \geq V_w - 1$. Hence, 
$$
RHS = \frac{p(V_w) - p(I_w)}{2p(V_w)} + \frac{p(I_w)}{p(V_w)} \sum_{e \in \Stem(w)} \frac{J_w^e}{V_w-1} \geq \frac{p(V_w) - p(I_w)}{2p(V_w)} + \frac{p(I_w)}{p(V_w)}  \geq \frac{1}{2},\quad \forall w\in \act^-_t.
$$

For any $w \in \act^+$, we only keep the first term in \eqref{eq:3}. This proves \eqref{eq:4} and hence completes the proof of the lemma.
\end{proof}

\begin{proof}[Proof of Lemma~\ref{lemma-lower bnd}]
We are now ready to provide a lower bound for $\mathop{\mathbb{E}}\limits_{h^*\sim D}[\st_t]$. Recall that $\st_t$ is a random variable from Definition~\ref{score}. We have:
\begin{align}
\mathop{\mathbb{E}}_{h^*\sim D}[\st_t] &= \Pr[ \alg < t] \cdot 0 + \Pr[\alg \geq t]\cdot  \mathop{\mathbb{E}} \left[\st_t| \alg \geq t \right] \nonumber \\
&= \Pr[\alg \geq t] \sum_{w \in \act_t} \mathrm{Score}(t,w)\frac{\Pr[\text{\{$w$ observed at $t$\} $\cap$ \{$\alg \geq t$\}}]}{\Pr[ \alg \geq t]} \nonumber \\
&= \sum_{w \in \act_t} \mathrm{Score}(t,w) \Pr[\text{$w$ observed at $t$}]  = \sum_{w \in \act_t} \mathrm{Score}(t,w)\cdot p(V_w) \label{eq:reduce to act}.
\end{align}

The first equality in $\eqref{eq:reduce to act}$ holds because if state $w$ is observed  at time $t$ then the greedy algorithm won't terminate before time $t$, i.e., $\text{\{$w$ observed at $t$\}} \subseteq \{\alg \geq t\}$. For the last equality, we observe that $p(V_w)$ is exactly the probability that $w$ is observed at time $t$. 
Now,
\begin{align}
\sum_{w \in \act_t} \mathrm{Score}(t,w)p(V_w) &= \sum_{w \in \act_t^+} \mathrm{Score}(t,w) p(V_w) + \sum_{w \in \act_t^-} \mathrm{Score}(t,w) p(V_w) \nonumber \\
&\stackrel{}{\geq} \frac{L}{2t} \left(\sum_{w \in \act_t^+} \frac{p(V_w) - p(I_w)}{p(V_w)} \cdot p(V_w) + \sum_{w \in \act_t^-}p(V_w)\right) \label{eq:5} \\
&= \frac{L}{2t} \left(\sum_{w \in \act_t^+} \left(p(V_w) - p(I_w)\right) + \sum_{w \in \act_t^-}p(V_w)\right) \nonumber \\
&= \frac{L}{2t} \left(\sum_{w \in \act_t} p(V_w) - \sum_{w \in \act_t^+} p(I_w) \right) \quad  \stackrel{}{\geq} \quad  \frac{L}{2t} (a_t - o_{\frac{t}{L}}) . \label{eq:6} 
\end{align}
The inequality in \eqref{eq:5} is by  Lemma~\ref{keylemma 1}. For the inequality in \eqref{eq:6}, we use two results:
  $\sum\limits_{w \in \act_t} p(V_w) = a_t$ by Lemma~\ref{partition} and   $\sum\limits_{w \in \act_t^+} p(I_w) \leq o_{t/L}$ by  Lemma~\ref{lemma:I_w}. 
\end{proof}

We now turn to proving an upper bound on the score, which is almost identical to a similar lemma in \cite{althani2024}.
We provide the proof for completeness. 
\begin{lemma}\label{lem:upper-bnd}
For any $t \geq 0$, we have the upper bound
$$
\int_{z\geq t} \quad \mathop{\mathbb{E}}_{h^* \sim D}[\st_z] dz \leq (1+ \ln m) a_t.
$$
\end{lemma}

\begin{proof}
For any state $w$, query $e$  and hypothesis $h\in H$, let $\Delta_h(e|w)=V_w-S_i(e,w)$ where $i=e(h)$ is the response of $e$ under $h$.  So, $\Delta(e|w)=\E_{h\sim D|w}[\Delta_h(e|w)]$.
We now express  the score in terms of the true hypothesis. For any time $z$ and hypothesis $h$, define
$$G_h(z) :=\left\{ \begin{array}{ll}
     0 & \mbox{if greedy terminates before time $z$ when true hypothesis is $h$} \\
     \frac{\Delta_h(e|w)}{c(e)(V_w-1)} & \substack{  \mbox{where $w\in \act_z$ is the state at time $z$ in greedy}\\ \mbox{ when true hypothesis is $h$, and $e$ is the query at state $w$} }  
\end{array} \right. . $$
Using Definition~\ref{score}, we obtain  $\E[\st_z] = \E_{h\sim D}[G_h(z)]$ for any time $z$. Hence,
\begin{align}
&\int_{z\geq t}  \E[\st_z] dz  = \int_{z\ge t} \E_{h^* }[ G_{h^*}(z)] dz = \E_{h^*}\left[\int_{z\ge t} G_{h^*}(z) dz\right] \notag\\
&= \Pr[\mbox{greedy terminates after time $t$}] \cdot \E_{h^*}\left[\int_{z\ge t} G_{h^*}(z) dz \big| \mbox{greedy terminates after time $t$}\right].\label{eq:score-ub-cases}
\end{align}
The equality in \eqref{eq:score-ub-cases} uses the fact that if greedy terminates before $t$ then $G_{h*}(z)=0$ for all $z\ge t$. 

  We will now show that 
  \begin{equation}\label{eq:harmonic}
\int_{z\ge 0} G_{h}(z) dz \le 1+\ln m,\quad \forall h\in H.      
  \end{equation}
  Fix any $h\in H$. 
Let $e_1, \dots, e_n$ be the query sequence in the greedy tree conditioned on $h^* = h$. For each $i\in [n]$   let $w_i$ be the state at which query $e_i$ is performed; let  $w_{n+1}$ be the state associated with the leaf node for $h$. The time interval during which state $w_i$ is active is $[\sum_{j=1}^{i-1} c(e_j) , \sum_{j=1}^{i} c(e_j)]$; note that the width of this interval  is $c(e_i)$. Moreover, $G_h(z)$ is a step function w.r.t.  $z$, and 
$$\int_{z\ge 0} G_{h}(z) dz =\sum_{i=1}^n c(e_i)\cdot  \frac{\Delta_h(e_i|w_i)}{c(e_i)(V_{w_i}-1)}  =  \sum_{i=1}^n \frac{V_{w_i} - V_{w_{i+1}}}{V_{w_i}-1}  \leq  1 + \ln m.
$$
To justify the last inequality, we first observe that $\{V_{w_i}\}_{i=1}^{n+1}$ is a strictly decreasing sequence of integers with values in $[1,m]$. We then use the following fact.
\begin{quote}
If  $\{a_i\}_{i=1}^{n+1}$ is a strictly decreasing sequence of integers with values in $[1, m]$ then: 
\[
\sum_{i=1}^{n} \frac{a_i - a_{i+1}}{a_i - 1} \le 1+ \ln m.
\]
\end{quote}
The proof is now complete by combining \eqref{eq:score-ub-cases} and \eqref{eq:harmonic}.
\end{proof}

We are ready to prove our main theorem.
\begin{proof}[Proof of Theorem~\ref{thm:main}]
Recall $\E\left[\alg\right] =\int_0^{\infty} a_t dt $ and $\E\left[\opt\right]=\int_0^{\infty} o_t dt$ from Equation~\eqref{eq:ALG and at}. So,
$$
\begin{aligned}
   & (1+ \ln m)\E\left[\alg\right]\stackrel{\text{(Eq.~\eqref{eq:ALG and at})}}{=} (1 + \ln m ) \int_0^{\infty} a_t dt \stackrel{\text{(Lemma~\ref{lem:upper-bnd})}}{\geq} \int_0^{\infty}\int_{z\geq t} \E_{h^* \sim D}[\st_z] dzdt \\
   &\stackrel{(\text{Lemma~\ref{lemma-lower bnd}})}{\geq} \int_0^{\infty} \int_{z\geq t} \frac{L}{2z} (a_z - o_{\frac{z}{L}}) dz dt  \stackrel{(\text{Fubini's Theorem})}{=}  \int_0^{\infty} \int_0^z  \frac{L}{2z} (a_z - o_{\frac{z}{L}}) dt dz \\
   &= \int_0^{\infty} \frac{L}{2} (a_z -o_{\frac{z}{L}} )dz \stackrel{(\text{Equation~\eqref{eq:ALG and at}})}{=}  \frac{L}{2} (\E\left[\alg\right] -  L \cdot \E\left[\opt\right]).
\end{aligned}
$$
Now, choosing $L = 4\cdot(1+\ln m)$, we get  $8 \cdot (1+\ln m) \E\left[\opt\right] \geq \E\left[\alg\right]$.
\end{proof}

\bibliographystyle{alpha}
\bibliography{references}

\end{document}